%% file: epstein-zin.tex
\documentclass{article}

\include{pensions-macros}

\title{Collectivised Pension Investment with Homogeneous Epstein--Zin Preferences}
\author{John Armstrong, Cristin Buescu}

\begin{document}
	
\maketitle

\begin{abstract}
In a collectivised pension fund, investors agree that any money remaining in the fund
when they die can be shared among the survivors.

We compute analytically the optimal investment-consumption
strategy for a fund of $n$ identical investors with homogeneous Epstein--Zin preferences,
investing in the Black--Scholes market in continuous time but consuming in discrete time.
Our result holds for arbitrary mortality distributions.

We also compute the optimal strategy for an infinite fund of investors, and prove 
the convergence of the optimal strategy as $n\to \infty$. The proof of convergence shows
that effective strategies for inhomogeneous funds can be obtained using the optimal strategies
found in this paper for homogeneous funds, using the results of \cite{ab-main}.

We find that a constant consumption strategy is suboptimal even for
infinite collectives investing in markets where assets provide no return
so long as investors are ``satisfaction risk-averse.'' This suggests that annuities
and defined benefit investments will always be suboptimal investments.

We present numerical results examining the importance of the fund size, $n$, and the market
parameters.
\end{abstract}

\section*{Introduction}

A group of individuals may group together and invest income for their retirement
in a collective fund. When an individual dies, any funds associated with that individual
are then divided among the survivors. The paper \cite{ab-main} shows how to model the
management of these funds mathematically and argues that they should yield significantly
better results for investors than traditional pension investment models.

This paper complements \cite{ab-main} by computing the optimal investment strategy for collective investment under the following assumptions:
\begin{enumerate}[(i)]
	\item  There are $n$ identical investors in the collective. \label{assumption:homogeneous}
	\item  The fund may invest in continuous time in a Black--Scholes--Merton market with one risky asset. \label{assumption:blackScholes}
	\item  The mortality of the individuals is independent: that is there is no systematic longevity risk.
		   Mortality occurs with a known probability distribution.
	\item  \label{assumption:epsteinZin} The preferences of each individual are given by homogeneous Epstein--Zin preferences
		   with mortality (as defined in \cite{ab-main}).
	\item  Consumption occurs in discrete time (for example once per year). \label{assumption:discreteTime}		   	
\end{enumerate}
We formulate the problem mathematically and are able to give analytical formulae for the optimal consumption and investment at
each time.

Assumption \eqref{assumption:homogeneous}, that the investors are identical, is not a significant restriction. The
paper \cite{ab-main} shows that once one knows how to manage a fund of identical individuals, it
is easy to devise very effective management strategies for {\em inhomogeneous funds}, i.e.\ funds of diverse individuals.

Assumption \eqref{assumption:blackScholes}, that the market is a Black--Scholes--Merton market is, of course, restrictive.
This is the simplest type of market we could consider. We are trading full realism for analytic tractability. The assumption that the market has only one risky asset is not restrictive. It follows from the mutual fund theorem arguments of \cite{armstrongClassification} that essentially the same strategy can be used in a Black--Scholes--Merton market with $n$ stocks.

Assumption \eqref{assumption:epsteinZin}, that the preferences are given by homogeneous Epstein--Zin utility, is a central
assumption to this paper, and is key to the analytic tractability of the problem. We consider
different possible models for preferences over consumption with mortality in \cite{ab-main}
and find that two models stand out as having particularly attractive properties. These
models are called {\em homogeneous Epstein--Zin preferences} and {\em exponential Kihlstrom--Mirman preferences} in \cite{ab-main} and we will use the same terminology. As the results of this
paper demonstrate, homogeneous Epstein--Zin preferences have the tremendous advantage
of being analytically tractable. This stems from the homogeneity property of these preferences. By contrast, we can only expect to solve analogous problems with exponential Kihlstrom--Mirman preferences using numerical methods (we show how to do this in \cite{ab-exponential}). It is 
possible to define inhomogeneous Epstein--Zin preferences with mortality (see \cite{ab-main}). We believe that the techniques of \cite{ab-exponential} could be applied to such problems, but that one cannot, in general, expect analytical results without homogeneity.

Assumption \eqref{assumption:discreteTime}, that consumption occurs in discrete time is not a significant restriction. Indeed one might argue it adds to the realism of the model.

The apparent mismatch between discrete time consumption  and continuous time investment
is the key technical trick required to obtain our analytic results. As is explained in \cite{armstrongClassification}, the Black--Scholes--Merton market is isomorphic to a linear
market in continuous time and this explains the analytic tractability of many problems involving this
market. However, in discrete time the Black--Scholes--Merton market is fundamentally non-linear, and
this explains the analytic intractability of problems such as Merton's investment problem in discrete time. We deduce that we must allow continuous time investment to obtain
analytic results.

On the other hand, if we are to allow arbitrary mortality distributions and continuous time consumption, then there is no hope of obtaining analytic results as one wouldn't even be able to write down the mortality distribution in general. Our assumption of discrete time consumption is essentially equivalent to assuming that mortality occurs in discrete time, and so restricts the set of mortality distributions we are considering to ensure tractability.

\medskip

The advantage of analytic tractability is the insight it gives us into optimal
pension investment.

For example, we can analyse how consumption varies over time. We find that except for very special cases, constant consumption is never optimal.
This is interesting because many people see a defined-benefit pension which provides constant real-terms income as the ``gold standard'' for a pension fund.
Our result shows that chasing constant income is, in fact, suboptimal.

It perhaps isn't so surprising that if market returns are non-zero there are advantages to taking some risk by investing in equities.

Nor should it be surprising that delaying consumption to benefit from market returns can also be advantageous. We are able to make this
precise by computing the elasticity of intertemporal substitution in our model.

It is perhaps, surprising that even if one assumes that both equities and bonds provide no return, it is still not optimal to receive a constant income,
if one is satisfaction risk-averse (see below for a definition). Under these circumstances it can be optimal to spend earlier (if the primary risk one perceives is the risk of
dying before one can consume one's pension) or to spend later (if the primary risk one perceives is the risk of living for a long time on an
inadequate pension).

\medskip

Let us now describe the structure of the paper.

Section \ref{sec:epsteinZinDef} reviews the definition of homogeneous Epstein--Zin preferences
with mortality.

Section \ref{sec:optimalInvestmentProblem} states the optimal investment problems we will solve.
The problem depends upon the number of individuals $n$. We will also state an investment problem
for a fund which is intended to represent the limiting case $n=\infty$.

Section \ref{sec:epsteinZinFormulae} solves the optimal investment problems analytically in the cases where
$n=1$ and $n=\infty$.

Section \ref{sec:consumption} computes how consumption and wealth vary over time, giving analytic
descriptions of their probability distributions.

Section \ref{sec:epsteinZinFiniteCollective} generalizes the results to arbitrary fund sizes $n$.

Section \ref{sec:convergence} uses our results to provide a rigorous justification
for the claim that our models for finite $n$ converge to the case $n=\infty$. This is of
obvious theoretical interest in its own right, but we remark that the proof
is essential to demonstrating that the strategies for inhomogeneous funds of described
in \cite{ab-main} will be effective.

\section{Homogeneous Epstein--Zin utility with mortality}
\label{sec:epsteinZinDef}

Let us recall the definition of homogeneous Epstein--Zin utility with mortality
given in \cite{ab-main}. In order to give a crisp
definition, we first define a convention for how we will we handle
algebra using infinite and infinitesimal values.

\begin{definition}
	The {\em extended positive reals} $\R^{++}$	is the set
	\[
	\R^{++} = \R^+ \cup \{\epsilon^\alpha \mid \alpha \in \R \setminus \{0\} \}
	\]
	where $\epsilon$ is a symbol representing an infinitesimal value. We extend addition, multiplication and raising to a real power to $\R^{++}$ in the obvious way:
	\begin{align*}
	x + \epsilon^\alpha &= \begin{cases}
	x & \text{when }\alpha > 0 \\
	\epsilon^\alpha & \text{otherwise}
	\end{cases} \\
	\epsilon^\alpha + \epsilon^\beta &= \epsilon^{\min\{\alpha,\beta\}} \\
	x \epsilon^\alpha &= \epsilon^\alpha \\
	\epsilon^\alpha \epsilon^\beta &= \epsilon^{\alpha + \beta} \\
	(\epsilon^\alpha)^\beta &= \epsilon^{\alpha \beta}
	\end{align*}
\end{definition}

We now assume that we are given a time grid ${\cal T}=\{t_0,t_0+\delta t,t_0+2\, \delta t,
t_0+3\, \delta t
\ldots,T-\delta t\}$ where $t_0$ is some initial time, $\delta t$ is a fixed time
step and $T$ is the time by which we assume mortality is certain.

We will model the individual's consumption as a non-negative stochastic process $(\gamma_t)_{t \in \cal T}$
in a filtered probability space $(\Omega, {\cal F}, {\cal F}_t, \P)$.
The time of death $\tau$ is a stopping time taking values in ${\cal T}$.
Our convention is that any consumption up to {\em and including} time $\tau$ may effect
the individual's utility, but any consumption occurring after time $\tau$ will be ignored.

\begin{definition}
	\label{def:homogeneousEpsteinZin}	
	{\em Homogeneous Epstein--Zin utility with mortality} is defined
	for a non-negative consumption process $(\gamma_t)_{t \in {\cal T}}$ and a stopping
	time $\tau$ taking values in ${\cal T}$. It depends on parameters $\alpha \in (-\infty,1) \setminus\{0\}$, $\rho \in (-\infty,1) \setminus\{0\}$,
	and $0<\beta = e^{-bt} \leq 1$.
	It is the $\R^{++}$-valued random process
	defined recursively by
	\begin{equation}
	Z_t(\gamma, \tau) =
	\begin{cases}
	\epsilon^{\frac{1}{\alpha}} & t > \tau; \\
	\left[ \gamma_t^\rho + \beta \, \E_t( Z_{t+\delta t}(\gamma, \tau)^\alpha )^\frac{\rho}{\alpha} \right]^\frac{1}{\rho} & \text{otherwise}.
	\end{cases}
	\label{eq:defepsteinzin}
	\end{equation}
\end{definition}

\section{The optimal investment problem}
\label{sec:optimalInvestmentProblem}

A general optimal investment problem for a homogeneous collective fund was described
mathematically in \cite{ab-main}. In this section we summarize the formulation
of \cite{ab-main}, specializing to the case of interest for this paper.

We model a collective fund of $n$ investors. The fund may invest in either a riskless bond
which grows at a risk-free rate of $r$ or in a stock whose price at time $t$, is denoted $S_t$.
The stock price at time $t_0$ is given. At subsequent times $S_t$ obeys
the SDE
\begin{equation}
\ed S_t = S_t(\mu \, \ed t + \sigma \, \ed W_t) \qquad S_{t_0}
\label{eq:blackScholesMerton}
\end{equation}
for a constant drift $\mu$ and volatility $\sigma$, and a $1$-dimensional Brownian motion $W_t$.

We will model consumption taking place in discrete time on a grid ${\cal T}$
as described in the previous section.
Since we are modelling consumption in discrete time, we may safely model mortality
in discrete time. We let $\tau$ be a random variable modelling the time of death of
a representative individual. We assume that $\tau$ has a probability distribution
given by $p_t \DT$ where $\DT$ is the measure given by adding the Dirac measures
associated to the grid points of ${\cal T}$.

Let $s_t$ denote the survival probability between times $t$ and $t+\delta t$. That is
\begin{equation}
s_t =
\frac{\sum_{i=1}^\infty p_{t + i \delta t}}{\sum_{i=0}^\infty p_{t + i \delta t}}
\label{eq:defs}
\end{equation}
Write $n_t$ for the number of individuals
individuals with a time of death greater than or equal to $t$. The process $(n_t)_{t \in {\cal T}}$ is a Markov process, with initial value $n_{t_0}$. Note, however, that $n_{t+\delta t}$ will be ${\cal F}_t$ measurable. We choose this convention for $n_t$ as it works well with our existing convention that individuals who die at age $t$ still consume at time $t$.

For the case of a finite number of individuals,
the transition probability of $n_t$ moving from
a value of $n$ at time $t$ to the value $i$ at time $t+\delta t$ is
given by
\begin{equation}
S_t(n,i):= \binom{n}{i} (s_t)^i (1-s_t)^{n-i}.
\label{eq:defS}
\end{equation}

We also wish to write down a formal optimal investment problem for the case
of a fund with $n=\infty$ investors. In this case we will define $n_t=\infty$
for all times up to $T$.

We will write $a_t$ for proportion of the fund invested in stock at time $t$.
We will write $X_t$ for the value of the fund per survivor at time
$t$ before consumption.
We will define $X_t=0$ if $n_t=0$.
Similarly, we will write ${\overline{X}}_t$ for the value of the fund per
survivor after consumption. We note that $X_t=\lim_{h\nearrow t} \overline{X}_h$ at time points $t \in {\cal T}$. At intermediate
times, $t \in [i\delta t, (i+1)\delta t)$, ${\overline{X}}_t$
obeys the SDE
\[
\ed \overline{X}_t = \overline{X}_t( a_t \mu + (1-a_t) r) \, \ed t
+ \overline{X}_t a_t \sigma \, \ed W_t,
\]
with initial condition given by the budget equation
\begin{equation}
\overline{X}_t = \begin{cases}
\frac{n_{t}}{n_{t+\delta t}} (X_t - \gamma_t) & n \text{ finite} \\
s_t^{-1} (X_t - \gamma_t) & n=\infty
\end{cases}
\label{eq:budgetEquationGeneral}
\end{equation}
unless $n_{t+\delta t}=0$ in which case we define $\overline{X}_t$ to be zero on $[t, t+\delta t)$. Note that this formula
is based on our convention that an individual who dies at a time $t$ still consumes at that time and the corresponding convention for $n_t$ which ensures $n_{t+\delta t}$ is ${\cal F}_t$ measurable.

Let $\tau^i$ denote the time of death of individual $i$.

For finite $n$ we let $(\Omega, {\cal F}, {\cal F}_t, \P)$ be the filtered probability space generated by
$W_t$ and the time of death variables $\tau^i$. We define $\tau=\tau^1$ to be the time of death
variable for one specific individual whose time of death is greater than or equal to $t_0$.

For $n=\infty$ we let $\tau$ be any random variable with the distribution $p_t \DT$.
We let $(\Omega,{\cal F}, {\cal F}_t,\P)$ be the filtered probability space generated by $W_t$ and $\tau$.

Let $\tilde{\cal A}(x,t_0)$ denote the space of admissible controls $(\gamma_t,a_t)$: that is ${\cal F}_t$-predictable processes such that
$0 \leq \gamma_t \leq X_t$ and with $X_{t_0}=x$. We define the value
function of our problem starting at time $t_0$ to be
\begin{equation}
v_n(x,t_0) = \sup_{(\gamma,a) \in \tilde{\cal A}(x,t_0)} Z_{t_0}(\gamma, \tau^1).
\label{eq:defValueFunction}
\end{equation}
where $Z_{t_0}$ is an Epstein--Zin utility function.

\section{The cases $n=1$ and $n=\infty$}
\label{sec:epsteinZinFormulae}

To highlight the key ideas we will consider
only the case when $n=1$ and $n=\infty$ in this section,
leaving the case of general $n$ until Section \ref{sec:epsteinZinFiniteCollective}.
We define $C$ to distinguish these cases as follows
\begin{equation}
C =
\begin{cases}
0 & n=1 \\
1 & n=\infty.
\end{cases}
\end{equation}

We write $z_t:=v_n(1,t)$ so
that the positive-homogeneity of $Z_t$ implies that
\begin{equation}
v_n(x,t)=x\, z_t.
\label{eq:zscaling}
\end{equation}
We note that we have not yet shown whether $z$ is finite, but equation \eqref{eq:zscaling} can still be interpreted for infinite values of $z$.

Let $c_t$ denote the consumption rate at time $t$, so $\gamma_t=c_t X_t$
for an individual who is still alive at time $t$.
The budget equation \eqref{eq:budgetEquationGeneral} 
can then be written as
\begin{equation}
\overline{X}_t = s_t^{-C} (1-c_t) X_t.
\label{eq:budgetEquation}
\end{equation}

We now let ${\cal A}_{t+{\delta t}}(x,c_t)$ be the set of random variables
$X_{t+\delta t}$ representing the value of our investments at time $t+\delta t$
that can be obtained by a continuous time trading strategy with
an initial budget given by \eqref{eq:budgetEquation} with $X_t=x$. Thus
${\cal A}_{t+\delta t}(x,c_t)$ is the set of admissible investment returns given the budget
and the consumption.

The Markovianity of Epstein--Zin preferences and the definition of $Z_t$ to compute allow us to apply the dynamic programming principle to compute
\begin{align}
v_n(x,t) &= \sup_{c_t \geq 0, X_{t+\delta t}\in {\cal A}_{t+\delta t}(x,c_t)}
[
(x\, c_t)^\rho
+ \beta \{s_t \E_t( v_n(X_{t+\delta t}, t+\delta t)^\alpha )\}^\frac{\rho}{\alpha}
]^\frac{1}{\rho} \nonumber \\
&= \sup_{c_t \geq 0, X_{t+\delta t}\in {\cal A}_{t+\delta t}(x,c_t)}
[
(x\, c_t)^\rho
+ \beta z_{t+\delta t}^\rho  \{ s_t \E_t( X_{t+\delta t}^\alpha )\}^\frac{\rho}{\alpha}
]^\frac{1}{\rho}.
\label{eq:valueFunction1}
\end{align}
The second line follows from the first by the positive homogeneity of Epstein--Zin utility \eqref{eq:zscaling}.

If $\alpha>0$, we may compute the value of 
\[
\sup_{X_t \in {\cal A}_{t+\delta t}(c_t)} \E_t( X_{t+\delta t}^\alpha )
\]
by solving the Merton problem for optimal investment over time
period $[t, t+ \delta t]$, with initial budget $\overline{X}_t$, no consumption,
and terminal utility function $u(x)=x^\alpha$.

We find
\begin{equation}
\sup_{X_t \in {\cal A}_{t+\delta t}(c_t)} \E_t( X_{t+\delta t}^\alpha ) = (\exp( \xi \, \delta t) \overline{X}_t)^\alpha
\label{eq:mertonOptimum}
\end{equation}
where
\begin{align}
\xi &= \sup_{a \in \R}[ a(\mu - r) + r - \frac{1}{2}a^2(1-\alpha)\sigma^2]
\nonumber \\
&= a^*(\mu - r) + r - \frac{1}{2}(a^*)^2(1-\alpha)\sigma^2, \quad \text{with } a^*:= \frac{\mu-r}{(1-\alpha) \sigma^2}
\label{eq:defastar}
\end{align}
For details see Merton's paper \cite{merton1969lifetime} or \cite{pham} equations (3.47) and (3.48).  Moreover the proportion invested in stock
is given by $a^*$ which is a constant determined entirely by $\alpha$ and the market.
In the case where $\alpha<0$ we must instead compute
\[
\inf_{X_t \in {\cal A}_{t+\delta t}(c_t)} \E_t( X_{t+\delta t}^\alpha ).
\]
However, apart from the change of a $\sup$ to an $\inf$ everything is algebraically identical, so the same formulae emerge.

Putting the value \eqref{eq:mertonOptimum} into our
expression \eqref{eq:valueFunction1} for the value function we obtain
\begin{align}
v_n(x,t) &= 
\sup_{c_t\geq 0}
[
(x\, c_t)^\rho
+z_{t+\delta t}^\rho \beta \, (s_t \exp( \alpha \xi \, \delta t) \overline{X}_t^\alpha)^\frac{\rho}{\alpha} ]^\frac{1}{\rho}
\nonumber \\
&= 
\sup_{c_t\geq 0}
[
(x\, c_t)^\rho
+\beta \, (z_{t+\delta t} \exp( \xi \, \delta t) s_t^{(\frac{1}{\alpha}-C)} (1-c_t) X_t)^{\rho} ]^\frac{1}{\rho}
\label{eq:valueFunction2}
\end{align}
where the last line follows from the budget equation \eqref{eq:budgetEquation}.
We define
\begin{equation}
\phi_t:=\beta^\frac{1}{\rho} \exp(\xi \delta t) s_t^{(\frac{1}{\alpha}-C)}
\quad \theta_t:=\phi_t z_{t+\delta t}
\label{eq:defphi}
\end{equation}
so that \eqref{eq:valueFunction2} may be written as
\begin{align}
z_t &= 
\sup_{c_t\geq 0}
[
(c_t)^\rho
+\theta_t^\rho (1-c_t)^\rho ]^\frac{1}{\rho}
\label{eq:valueFunction3}
\end{align}

Differentiating the expression in square brackets on the right-hand
side, we see that the supremum is achieved in equation \eqref{eq:valueFunction2}
when $c_t=c^*_t$, where $c^*_t$ satisfies
\[
\rho \, (c^*_t)^{\rho-1} - \rho \theta_t^\rho((1-c^*_t)^{\rho-1} = 0,
\]
or, if this yields a negative value for $c^*_t$, we should take $c^*_t=0$.
Simplifying we must have
\begin{equation}
\left( \frac{c^*_t}{1-c^*_t} \right)^{\rho-1} = \theta_t^\rho.
\label{eq:cStarRelation}
\end{equation}
So
\begin{equation}
c^*_t = 
(1+\theta_t^\frac{\rho}{1-\rho})^{-1}.
\label{eq:defcstar}
\end{equation}
This expression for $c^*_t$ is non-negative, so it always gives the argument for the supremum  in \eqref{eq:valueFunction2}. We obtain
\begin{align}
z_t
&=
[(c^*_t)^\rho
+  \theta_t^\rho (1-c^*_t)^{\rho} ]^\frac{1}{\rho} \nonumber \\
&=
c^*_t \left[ 1
+  \theta_t^\rho \left(\frac{1-c^*_t}{c^*_t}\right)^{\rho} \right]^\frac{1}{\rho} \nonumber \\
&=
c^*_t \left[ 1
+  \theta_t^\rho \left(\theta_t^\frac{\rho}{\rho-1} \right)^{-\rho}
\right]^\frac{1}{\rho}, \quad \text{by \eqref{eq:cStarRelation}}, \nonumber \\
&=
c^*_t \left[ 1
+  (\theta_t^{1-\frac{\rho}{\rho-1}})^{\rho}
\right]^\frac{1}{\rho} \nonumber \\
&=
c^*_t \left[ 1
+  \theta_t^\frac{\rho}{1-\rho}
\right]^\frac{1}{\rho} \nonumber \\
&=
(1 + \theta_t^\frac{\rho}{\rho-1})^{-1} ( 1
+  \theta_t^\frac{\rho}{1-\rho}
)^\frac{1}{\rho}, \quad \text{ by \eqref{eq:defcstar}}, \nonumber \\
&=
( 1 +  \theta_t^\frac{\rho}{1-\rho})^\frac{1-\rho}{\rho}. \nonumber
\end{align}
We conclude that
\begin{equation}
z_t^\frac{\rho}{1-\rho} = 1 + \theta_t^\frac{\rho}{1-\rho} = 1 + \phi_t^\frac{\rho}{1-\rho} z_{t+\delta t}^\frac{\rho}{1-\rho}
\label{eq:valueFunction4}
\end{equation}
where $\phi$ is given by equation \eqref{eq:defphi}. We observe also that equation \eqref{eq:defcstar} for the optimal consumption rate per survivor simplifies to
\begin{equation}
c^*_t = z_t^\frac{\rho}{\rho-1}.
\label{eq:cStarConclusion}
\end{equation}

We summarize our findings below.
\begin{theorem}
	[Optimal investment with Epstein--Zin preferences]	
	Suppose that an individual has probability $p_t \delta t$ of dying
	in the interval $[t, t+ \delta t]$. Individuals consume at fixed time points $i\, \delta t$. By the time $N \delta t$, death is certain. Between time points, we may trade in the Black--Scholes--Merton
	market \eqref{eq:blackScholesMerton}. Let $C=0$ if we are interested
	in optimizing the consumption of an individual and $C=1$ if we are interested in the collectivised problem. The utility of each individual
	is given by Epstein--Zin utility of the form \eqref{eq:defepsteinzin}.
	Then
	\begin{enumerate}[(i)]
		\item  The optimal proportion of stock investments is
		determined entirely by the monetary-risk-aversion $\alpha$
		and the market. In particular it is independent of time
		and wealth. It is given by the value $a^*$ given in equation
		\eqref{eq:defastar}.
		\item  The optimal Epstein--Zin utility is given by $x z_t$ where 	$z_t$ obeys the
		difference equation \eqref{eq:valueFunction4}. The value
		of $\phi_t$ is given in equation \eqref{eq:defphi}.
		The optimal Epstein--Zin utility may be computed recursively
		since $z_{N \delta t}=0^\frac{1}{\alpha}$.	
		\item  The consumption for each survivor at time $t$ is given by $\gamma_t = x_t c^*_t$ where $c^*_t$ is given by \eqref{eq:cStarConclusion}.		   
	\end{enumerate}
	\label{thm:epsteinZin}
\end{theorem}

\section{Consumption over time}
\label{sec:consumption}

It is interesting to see how wealth and the consumption per individual vary over time.
\begin{theorem}
	Under the same conditions as \ref{thm:epsteinZin},
	the optimal fund value per survivor at time $t$, $X_t$, follows a log
	normal distribution. Write $\mu^X_t$ and $\sigma^X_t$ for
	the mean and standard deviation of $\log X_t$ so that
	\begin{equation}
	\log X_t \sim N( \mu^X_t, (\sigma^X_t)^2 ).
	\label{eq:distLogWealth}
	\end{equation}
	The standard deviation is given by
	\begin{equation}
	\sigma^X_t = \sigma a^* \sqrt{t}.
	\label{eq:sdLogWealth}
	\end{equation}
	where $a^*$ is given by \eqref{eq:defastar}.
	The mean satisfies the difference equation
	\begin{equation}
	\mu^X_{t+\delta t}=\mu^X_{t} + \log(s_t^{-C}) + 
	\log\left( 1 - z_t^\frac{\rho}{\rho-1} \right) + \tilde{\xi} \delta t, \qquad \mu^X_{0} = \log (x_0)
	\label{eq:meanLogWealth}
	\end{equation}
	where $z_t$ is given by \eqref{eq:valueFunction4} and where we define $\tilde{\xi}$ by the same formula
	used to define $\xi$ but with $\alpha$ set to zero, i.e.\
	\begin{equation}
	\tilde{\xi} := a^*(\mu - r) + r - \frac{1}{2}(a^*)^2 \sigma^2.
	\label{eq:defTildeXi}
	\end{equation}	
	The optimal consumption per survivor $\gamma_t$ at time $t$ is also
	log normally distributed with
	\begin{equation}
	\log \gamma_t \sim N( \tfrac{\rho}{\rho-1} \log(z_t) + \mu^X_t,
	(\sigma^X_t)^2 ).
	\label{eq:distConsumption}
	\end{equation}
	The mean of the log consumption per survivor satisfies the equation
	\begin{equation}
	\E( \log \gamma_{t+\delta t} \mid \gamma_t ) = 
	\log(\gamma_t) + \log(s_t^{-C}) + \frac{\rho}{1-\rho} \log(\phi_t) + \tilde{\xi} \delta t
	\label{eq:meanLogGammaDynamics}
	\end{equation}
	where $\phi_t$ is given by equation \eqref{eq:defphi}.
	\label{thm:epsteinZinConsumption}
\end{theorem}
\begin{proof}
	We suppose as induction hypothesis that the
	distribution of $X_t$ is as described at time $t$.
	
	The budget equation \eqref{eq:budgetEquation}
	then tells us that the wealth per survivor after consumption, ${\overline{X}}_t$, satisfies
	\[
	\overline{X}_t = s^{-C}_t (1- z_t^{\frac{\rho}{\rho-1}}) X_t.
	\]
	Hence
	\[
	\log \overline{X}_t \sim
	N( \mu^X_t + \log(s^{-C}_t (1- z_t^{\frac{\rho}{\rho-1}})),
	(\sigma^X_t)^2 ).
	\]
	Our investment strategy from $t$ to $(t+\delta t)$ is a continuous time trading strategy where we hold a fixed proportion of our wealth in stocks at all time. So, in the interval $(t,t+\delta t]$, $X_t$ satisfies the SDE
	\[
	\ed X_t = (1-a^*) r X_t \, \ed t + a^* X_t (\mu \, \ed t + \sigma \, \ed W_t), \qquad X_t=\overline{X}_t.
	\]
	By It\^o's lemma we find
	\begin{align}
	\ed (\log X)_t &= (1-a^*) r \, \ed t + a^* ((\mu-\tfrac{1}{2}a^* \sigma^2) \, \ed t + \sigma \, \ed W_t), \qquad \log X_t = \log \overline{X}_t \nonumber \\
	&= \tilde{\xi} \ed t  + a^* \sigma \, \ed W_t.
	\end{align}
	We deduce that $(\log X_{t+\delta t}-\log \overline{X}_t)$ 
	conditioned on the value of $X_t$ will follow a normal distribution with
	mean $\tilde{\xi} \, \delta t$ and standard deviation $a^* \sigma \sqrt{\delta t}$.
	Moreover the random variable $(\log X_{t+\delta t}-\log \overline{X}_t)$
	is independent of $\overline{X}_t$. 
	
	The sum of independent normally distributed random increments yields a new
	normally distributed random variable, and one can compute the mean and variance by adding the mean and variance of the increments. Hence
	\[
	\log X_{t+\delta t}
	\sim N( \mu^X_{t+\delta t}, (\sigma^X_{t+\delta t})^2 )
	\]
	where
	\begin{equation}
	\mu^X_{t+\delta t} = 
	{\mu}^X_t
	+ \log(s^{-C}_t (1- z_t^{\frac{\rho}{\rho-1}}))
	+ \tilde{\xi} \, \delta t
	\label{eq:muXRecursion}
	\end{equation}
	and
	\begin{equation}
	(\sigma^X_{t+\delta t})^2 = (\sigma^X_t)^2 + (a^*)^2 \sigma^2 \delta t  .
	\label{eq:sigmaXRecursion}
	\end{equation}
	Solving the recursion \eqref{eq:sigmaXRecursion}
	yields equation \eqref{eq:sdLogWealth}.
	The result for $X_t$ now follows by induction.
	
	Equation \eqref{eq:distConsumption} follows from equation
	\eqref{eq:cStarConclusion}. Using \eqref{eq:distConsumption}, \eqref{eq:meanLogWealth}
	we calculate
	\begin{align*}
	\E( \log&\gamma_{t+\delta t} \mid \gamma_t) - \log( \gamma_t )  \\
	&=
	\frac{\rho}{\rho-1} \left( \log(z_{t+\delta t})
	- \log(z_t) \right)
	+ \log( s_t^{-C}) + \log( 1 - z_t^\frac{\rho}{\rho-1}) + \tilde{\xi} \delta t
	\\
	&=
	\frac{\rho}{1-\rho}\left(
	\log(z_t) -\log(z_{t+\delta t})  \right)
	+ \log( s_t^{-C}) + \log\left( \frac{z_t^\frac{\rho}{1-\rho} - 1}{z_t^\frac{\rho}{1-\rho}} \right)
	+ \tilde{\xi} \delta t \\
	&= \log( s_t^{-C} ) +
	\log\left(  \phi_t^{\frac{\rho}{1-\rho}}\right) + \tilde{\xi} \delta t
	\qquad \text{ by equation \eqref{eq:valueFunction4}.}
	\end{align*}
	This completes the proof.
\end{proof}

To interpret Theorem \ref{thm:epsteinZinConsumption} we specialize to the case of a market where $\mu=r=0$ and 
to preferences with $\beta=1$. This represents the problem of consuming a
fixed lump sum over time when there is no inflation but also no investment opportunities. While not financially reasonable, this problem highlights
how longevity risk affects consumption, when considered in isolation from market risk.

In this case $\gamma_t$ is a deterministic function. We find from equations \eqref{eq:meanLogGammaDynamics}
and \eqref{eq:defphi} that
\begin{equation}
\gamma_{t+\delta t}=
\left(s_t^{\frac{1}{\alpha }-\frac{C}{\rho }}\right)^{\frac{\rho }{1-\rho }} \gamma_t.
\label{eq:gammaDynamics}
\end{equation}
We note that $s_t$ is a non-zero probability, so $0 < s_t \leq 1$. We may use equation \eqref{eq:gammaDynamics} to compute whether $\gamma$ increases or decreasing over time. We summarize the results in Table \ref{table:direction}.

Perhaps surprisingly, we find that sometimes consumption
increases over time rather than decreases. In the collectivised case, this can be explained by the fact that the pension will always be inadequate when $\alpha<0$. We say that a pension is inadequate
if living an extra year on that pension decreases utility. If $\alpha<0$, living longer is
considered negative by homogeneous Epstein--Zin preferences, so we may wish to compensate individuals who have the misfortune to live longer. We cannot identify these individuals in advance, so the only way to provide this compensation is to increase consumption over time. Thus the increasing consumption arises from the fact that when $\alpha<0$, the primary risk is the inadequacy of the pension, when $\alpha>0$, the primary risk is dying young. We believe that this mixing of the notion of pension adequacy with the risk aversion parameter is a significant shortcoming of homogeneous Epstein--Zin preferences. It is the price one must pay for analytic tractability.

In the individual case, the concern that one will die young is much more serious. This is why for the individual problem, the fear of an inadequate pension only dominates when both $\alpha<0$ and $0<\rho<1$.

The case when $\alpha=\rho$ corresponds to the case of standard
von-Neumann Morgernstern preferences, in which case constant consumption
is optimal in the collectivised case as was proved in \cite{ab-main}.

More significantly, our result also shows the converse. Constant consumption from one period to the next is only optimal if and only if either (i) the survival probability is one, or (ii) $\alpha=\rho$ so one is satisfaction risk-neutral. Hence, even ignoring market effects constant consumption will be suboptimal for any realistic parameter choices.

We have not shown the case $\alpha>\rho$ in Table \ref{table:direction} as in this case one has monetary-risk-aversion but not satisfaction-risk-aversion. We found the resulting behaviour to be difficult to interpret as rational, cautious (as understood intuitively) strategies. We see this simply as a evidence that satisfaction-risk-aversion is the correct operationalization of the intuitive notion of risk-aversion.

\begin{table}
	\begin{center}
		\begin{tabular}{ccc} \toprule
			& Collectivised & Individual \\ \midrule
			$\alpha<\rho<0<1$ & Increasing & Decreasing \\
			$\alpha<0<\rho<1$ & Increasing & Increasing \\ 
			$0<\alpha<\rho<1$ & Decreasing & Decreasing \\ 
			$\alpha=\rho<0<1$ & Constant & Decreasing \\
			$0<\alpha=\rho<1$ & Constant & Decreasing \\ \bottomrule
		\end{tabular}
		\caption{The behaviour of consumption with time when $\mu=r=0$, $\beta=1$ and $\alpha\leq\rho$.}
	\end{center}
	\label{table:direction}
\end{table}

\medskip

It is also interesting to calculate how consumption changes 
according to the available investment opportunities.
If interest rates increase one may choose to defer consumption to
a later date to benefit from the increased rate.
To quantify this behaviour one wishes to calculate
the elasticity of intertemporal substitution which
is defined as follows.

\begin{definition}
	The {\em elasticity of intertemporal substitution} at time $t$ is defined to be
	\[
	\EIS_t:= \frac{1}{\delta t}
	\frac{\ed}{\ed r} \left( \E( \log( \gamma_{t+1}) )
	- \log( \gamma_{t}) \right).
	\]
	When $\gamma_t$ is deterministic, this definition corresponds
	with the standard definition \cite{hall}.
\end{definition}

Theorem \ref{thm:epsteinZinConsumption} allows us
to calculate this elasticity.
\begin{corollary}
	\label{corollary:eis}	
	For the optimal investment strategy of Theorem \eqref{thm:epsteinZin}
	we have
	\[
	\EIS_t = \frac{1}{1 - \rho} \left( 1 - \frac{(\mu - r) (1 + \alpha (\rho - 2))}{(\alpha - 1)^2 \sigma^2}
	\right).
	\]
	If $\mu=r$ this simplifies to
	\[
	\EIS_t = \frac{1}{1 - \rho}.
	\]
	In the case of von Neumann-Morgernstern utility we have
	\[
	\EIS_t = \frac{1}{1 - \rho} \left( 1 - \frac{\mu - r}{\sigma^2} \right).
	\]
\end{corollary}
\begin{proof}
	From  \eqref{eq:meanLogGammaDynamics} and \eqref{eq:defphi} we immediately
	find
	\[
	\EIS_t = \frac{\ed }{\ed r} \left(\frac{\rho}{1-\rho} \xi - \tilde{\xi} \right).
	\]
	The result is now a simple calculation from \eqref{eq:defastar} and \eqref{eq:defTildeXi}.
\end{proof}

\section{General finite collectives}
\label{sec:epsteinZinFiniteCollective}

Let $v_n$ be the value function \eqref{eq:defValueFunction} for
the optimal investment problem for $n$ individuals with
homogeneous Epstein--Zin preferences. By positive homogeneity
we may define $z_{n,t}:=v_n(1,t)$, so that $v_n(x,t)=x z_{n,t}$.

Let $I_{i,t_0+\delta t}$ denote the event that both
\begin{enumerate}[(i)]
	\item there are $i$ survivors at time $t_0+\delta t$, i.e.\ $n_{t+\delta t}=i$.
	\item the individual whose utility we wish to calculate is one of those survivors, so $\tau^\iota>t$.
\end{enumerate}
Recall that $X_{t}$ denotes the value of the fund per survivor at time $t$ before consumption
and mortality.
\[
\E_{t_0}(Z_{t_0+\delta t}^\alpha \mid I_{i,{t_0+\delta t}}) = 
\E_{t_0}( v_i( X_{t_0+\delta t}, t_0 + \delta t)^\alpha \mid I_{i,t_0+\delta t}).
\]
Hence
\[
\E_{t_0}(Z_{t_0+\delta t}^\alpha) = \sum_{i=1}^{n_t} \frac{i}{n_t} S_t(n_t,i) \E_{t_0}(v_i( X_{t_0+\delta t}, t_0 + \delta t)^\alpha \mid I_{i,t_0+\delta t}).
\]
We now let ${\cal A}_{t+{\delta t}}(x,c_t)$ be the set of random variables
$X_{t+\delta t}$ representing the value of the fund per survivor at time $t+\delta t$
before consumption that can be obtained by a continuous time trading strategy given initial capital $\overline{X}_t = \frac{n_{t+\delta t}}{n_t} (1-c_t) X_t$ per survivor when $X_t=x$. Writing $c_t$ for the rate of consumption and using the dynamic programming principle we find
\begin{align}
v_n(x,t) &= \sup_{\stackrel{c_t \geq 0}{X_{t+\delta t}\in {\cal A}_{t+\delta t}(x,c_t)}}
\left[
(x\, c_t)^\rho
+ \beta \left\{\sum_{i=1}^n  \frac{i}{n} S_t(n,i) \E_t( v_i(X_{t+\delta t}, t+\delta t)^\alpha \mid I_{i,t_0+\delta t} ) \right\}^\frac{\rho}{\alpha}
\right]^\frac{1}{\rho} \nonumber \\
&= \sup_{\stackrel{c_t \geq 0}{X_{t+\delta t}\in {\cal A}_{t+\delta t}(x,c_t)}}
\left[
(x\, c_t)^\rho
+ \beta \left\{ \left(\sum_{i=1}^n \frac{i}{n} S_t(n,i) z_{i,{t+\delta t}}^\alpha \E_t(  X_{t+\delta t}^\alpha \mid I_{i,t_0+\delta t} ) \right)  \right\}^\frac{\rho}{\alpha}
\right]^\frac{1}{\rho}. \nonumber \\
\end{align}
We used positive homogeneity to obtain the last line. The argument of Section \ref{sec:epsteinZinFormulae} tells us how to optimize over $X_t$.
Hence we find
\begin{equation*}
z_{n,t} =
\sup_{c_t \geq 0}
\left[
(c_t)^\rho
+ \beta \left(\sum_{i=1}^n \left( \frac{i}{n} \right)^{1-\alpha} S_t(n,i) z_{i,t+\delta t}^{\alpha} \right)^\frac{\rho}{\alpha} \left(\exp( \xi \, \delta t) (1-c_t)\right)^\rho  
\right]^\frac{1}{\rho}
\end{equation*}
where $\xi$ is as defined in equation \eqref{eq:defastar}. The optimal investment policy is also described in equation \eqref{eq:defastar}, and as before it depends only on the market and the monetary-risk-aversion parameter $\alpha$.

We may rewrite our expression for $z_{n,t}$ as follows:
\begin{equation}
z_{n,t} =
\sup_{c_t \geq 0}
\left[
(c_t)^\rho
+ \tilde{\theta}_{n,t}^\rho (1-c_t)^\rho  
\right]^\frac{1}{\rho}
\label{eq:valueFunction3Tilde}
\end{equation}
where
\begin{equation}
\tilde{\theta}_{n,t} =  \beta^\frac{1}{\rho} \exp( \xi \, \delta t) \left(
\sum_{i=1}^n \left( \frac{i}{n} \right)^{1-\alpha} S_t(n,i) z_{i,t+\delta t}^{\alpha}
\right)^\frac{1}{\alpha}.
\label{eq:deftildetheta}
\end{equation}
Equation \eqref{eq:valueFunction3Tilde} is structurally identical to equation \eqref{eq:valueFunction3}. Hence from equation
\eqref{eq:valueFunction4} we may deduce similarly that
\begin{equation}
z_{n,t}^\frac{\rho}{1-\rho} = 1 + \tilde{\theta}_{n,t}^\frac{\rho}{1-\rho}.
\label{eq:valueFunction4Tilde}
\end{equation}

We state our results as a theorem.
\begin{theorem}
	Let $z_{n,t}$ denote the optimal Epstein--Zin utility, \eqref{eq:defValueFunction}, for a collective of $n$ individuals investing an amount $1$ at time $t$. The collective is
	allowed to invest in the Black--Scholes--Merton market \eqref{eq:blackScholesMerton} in continuous time. Individuals have independent mortality, with survival probability given by \eqref{eq:defs}.
	Then equations \eqref{eq:valueFunction4Tilde} and \eqref{eq:deftildetheta}
	together with the initial condition $z_{n,T}=0^\frac{1}{\alpha}$ allow us
	to compute the optimal Epstein--Zin utility by recursion. The value
	of $\xi$ is given in equation \eqref{eq:defastar} and the value
	of $S(n,i)$ is given in equation \eqref{eq:defS}
\end{theorem}

It is reassuring to check that $\tilde{\theta}_{1,t}$ in equation \eqref{eq:deftildetheta} coincides with the value of  $\theta_t$ for the individual problem given by equation \eqref{eq:defphi}.

\section{Convergence of $v_n$ as $n\to \infty$}
\label{sec:convergence}

In this section we will analyse the behaviour as $n \to \infty$ to prove the following.

\begin{theorem}
	\label{thm:cvgcEpsteinZin}	
	Let $z_{\infty,t}$ denote the maximum Epstein--Zin utility at time $t$
	for the infinitely collectivized case then
	\[
	z_{n,t} = z_{\infty,t} + O(n^{-\frac{1}{2}}).
	\]
\end{theorem}

Our proof strategy will be to approximate an expectation
involving the binomial distribution with a Gaussian integral which
we can then estimate using Laplace's method.
To get a precise convergence result, we need some estimates
on the rate of convergence
in the Central Limit Theorem. The estimates given in \cite{bhattacharya}
suit our purposes well. For the reader's convenience we will summarize the result we will need.

We begin with some definitions.
A random variable $X$ is said to satisfy Cram\'er's condition if its characteristic
function $f_X$ satisfies
\begin{equation}
\sup \{ |f_X(t)| \mid t> \eta \} < 1
\label{eq:cramer}
\end{equation}
for all positive $\eta$.
Let $\Phi$ be the standard normal
distribution. Given a set $A \subseteq \R$, and a function $g$, $\omega_g(A)$ is defined to equal
\[
\omega_g(A):=\sup\{|g(x)-g(y)| \mid x, y \in A \}.
\]
The set $B_{\epsilon}(x)$ is the ball of radius $\epsilon$ around $x$.

Let $Q_n$ be the appropriately normalized $n$-th partial sum
of a sequence of independent identically
distributed random variables $X^{(i)}$ for which Cram\'er's condition
holds and which have finite moments have
all orders. Appropriately normalized
means normalized such that the central
limit theorem implies $Q_n$ converges
to $\Phi$ in distribution.
Then there exists a constant $c$ such that for any bounded measurable function $g$
\begin{equation}
|\int_\R g \, \ed(Q_n - \Phi)| \leq c\, \omega_g(\R) n^{-\frac{1}{2}}
+ \int \omega_g(B_{c n^{-\frac{1}{2}}} (x)) \ed \Phi(x).
\label{eq:bhattacharya}
\end{equation}

The full result given in \cite{bhattacharya}
is more general and more precise than we need.
Let us explain how the statements are related.
We have simplified our statement to the
one-dimensional case, we have assumed the $X^{(i)}$
are identically distributed and we have assumed all 
moments of $X^{(1)}$ exist. The statement
in \cite{bhattacharya} is therefore more
complex, and in particular involves
additional terms called $\rho_{s,n}$
defined in the one-dimensional case by
\[
\rho_{s,n}:=\frac{1}{n}(\sum_{i=1}^n E|\sigma_n X^{(i)}|^s)
\]
where
\[
\sigma_n^2 := n \left( \sum_{i=1}^n \Var( X^{(i)} ) \right)^{-1}.
\]
Our assumptions on $X$ guarantee that $\rho_{s,n}$ is independent of $n$ and so
we have been able to absorb these terms into
the constant $c$. In addition, our statement
uses Theorem 1 of \cite{bhattacharya},
together with remarks at the end of the second paragraph on page 242 about
Cram\'er's condition.

\bigskip

We are now ready to prove the desired convergence result.

\begin{proof}[Proof of Theorem \ref{thm:cvgcEpsteinZin}]
	We proceed by a backward induction on $t$. The result is trivial for the case $t=T$. We now assume the induction hypothesis
	\[
	z_{n,t+\delta t} = z_{\infty,t+\delta t} + O(n^{-\frac{1}{2}}).	
	\]
	We wish to compute $\tilde{\theta}_{n,t}$, but only the sum
	in the expression \eqref{eq:deftildetheta} is difficult to compute. We
	will call this sum $\lambda_{n,t}$, so
	\begin{equation}
	\lambda_{n,t}:=\sum_{i=1}^n \left( \frac{i}{n} \right)^{1-\alpha} S_t(n,i) z^{\alpha}_{i, t+\delta t}.
	\label{eq:defnlambdant}
	\end{equation}
	
	Heuristically, one can approximate with a Gaussian integral using the Central Limit Theorem and then apply Laplace's method to compute the limit as $n \to \infty$. This motivates the idea of decomposing the sum above into a ``left tail'' for small values of $i$, a central term for values of $i$ near the mean of the Binomial distribution $n p$, and a ``right tail'' for larger values of $i$. We will in fact bound the tails separately (Steps 1 and 2, below),
	and then we will be able to rigorously apply a Central Limit Theorem
	argument to the central term (Step 3).
	
	We compute
	\begin{equation}
	\frac{S_t(n,i-1)}{S_t(n,i)} = \frac{i s_t}{(1 - i + n) (1 - s_t)}.
	\label{eq:exponentialDecay}
	\end{equation}
	We note that
	\[
	i \leq \frac{(n+1) (1-s_t)}{(\lambda -1) s_t+1} \implies \frac{i s_t}{(1 - i + n) (1 - s_t)} \leq \frac{1}{\lambda} \implies \frac{S_t(i-1,n)}{S_t(i,n)} \leq \frac{1}{\lambda}.
	\]
	So we define an integer $N_{\lambda, t, n}$ by
	\[
	N_{\lambda, t, n} := \left\lfloor \frac{(n+1) (1-s_t)}{(\lambda -1) s_t+1} \right\rfloor
	\]
	and then equation \eqref{eq:exponentialDecay} will ensure that we
	have exponential decay of $S_t(i,n)$ as $i$ decreases
	\begin{equation}
	i\leq N_{\lambda, t, n} \implies S(n,i) \leq \lambda^{i-N_{\lambda,t, n}}
	S(n,N_{\lambda, t, n}).
	\label{eq:exponentialdecay}
	\end{equation}
	{\bf Step 1.} We can now estimate the left tail of \eqref{eq:defnlambdant}. There
	exists a constants $C_{1,t}$, $C_{2,t}$ such that
	\begin{align}
	\sum_{i=1}^{N_{3,t,n}}  \left( \frac{i}{n} \right)^{1-\alpha} S_t(n,i)
	&\leq C_{1,t} \int_1^{N_{3,t,n}} 3^{-N_{3,t,n}+i} S(n,N_{3, t, n}) \left(\left(\frac{1}{n}\right)^{1-\alpha}+1 \right)
	\, \ed i \nonumber \\
	&\leq C_{2,t} S(n,N_{3, t, n}) \left(\left(\frac{1}{n}\right)^{1-\alpha}+1 \right).
	\label{eq:lefttail1}
	\end{align}
	To estimate this term, we observe that
	\begin{align*}
	N_{2,t,n}-N_{3,t,n} &= \left \lfloor 
	\frac{(n+1) (1-s_t)}{s_t+1}
	\right \rfloor
	-
	\left \lfloor \frac{(n+1) (1-s_t)}{2 s_t+1}
	\right \rfloor \\
	&\geq \left \lfloor 
	\frac{(n+1) (1-s_t)}{s_t+1}
	-
	\frac{(n+1) (1-s_t)}{2 s_t+1}
	\right \rfloor -2 \\
	&= \left \lfloor 
	\frac{(n+1) s_t (1-s_t)}{(s_t+1)(2 s_t + 1)}
	\right \rfloor -2
	\end{align*}
	Hence by equations \eqref{eq:exponentialDecay} and \eqref{eq:lefttail1} we
	find
	\begin{align}
	\sum_{i=1}^{N_{3,t,n}}  \left( \frac{i}{n} \right)^{1-\alpha} S_t(n,i)
	&\leq C_{2,t} S(N_{2, t, n},n) 2^{-\left \lfloor 
		\frac{(n+1) s_t (1-s_t)}{(s_t+1)(2 s_t + 1)}
		\right \rfloor +2} \left(\left(\frac{1}{n}\right)^{1-\alpha}+1 \right) \nonumber \\
	&\leq C_{2,t} 2^{-\left \lfloor 
		\frac{(n+1) s_t (1-s_t)}{(s_t+1)(2 s_t + 1)}
		\right \rfloor +2} \left(\left(\frac{1}{n}\right)^{1-\alpha}+1 \right). \nonumber
	\end{align}
	which decays exponentially as $n\to \infty$. Our induction hypothesis
	ensures that the $z_{i,t+\delta t}^\alpha$ are bounded, so we may safely
	conclude that
	\begin{equation}
	\lambda_{n,t}:=\left( \sum_{i=N_{3,t,n}}^n \left( \frac{i}{n} \right)^{1-\alpha} S_t(n,i) z^{\alpha}_{i, t+\delta t} \right)
	+ O(n^{-\frac{1}{2}})
	\end{equation}
	{\bf Step 2.} We apply the same strategy to the right tail. This time we compute
	\[
	\frac{S_t(n,i+1)}{S_t(n,i)}=\frac{(1-s_t) (n-i)}{(i+1) s_t}
	\]
	we define
	\[
	M_{\lambda,t,n} = \left\lceil \frac{\lambda  n (1-s_t)-s_t}{\lambda (1-s_t)+s_t} \right\rceil
	\]
	to ensure that
	\[
	i \geq M_{\lambda, t, n} \implies  \frac{S_t(n,i+1)}{S_t(n,i)} \leq \frac{1}{\lambda}.
	\]
	Repeating the same argument as for the left tail tells us that
	\begin{equation}
	\lambda_{n,t}:=\left( \sum_{i=N_{3,t,n}}^{M_{3,t,n}} \left( \frac{i}{n} \right)^{1-\alpha} S_t(n,i) z^{\alpha}_{i, t+\delta t} \right)
	+ O(n^{-\frac{1}{2}})
	\label{eq:lambdamiddle}
	\end{equation}
	We note that $\left(\frac{i}{n}\right)^{1-\alpha}$ is monotonic in
	$i$ and that  $\frac{N_{3,t,n}}{n}$ and $\frac{N_{3,t,n}}{n}$ tend
	to finite, non-zero limits as $n\to \infty$. We deduce that
	there exists a constant $C_{3,t}$ such that
	\begin{equation}
	N_{3,t,n}\leq i \leq M_{3,t,n} \implies \left|\left( \frac{i}{n} \right)^{1-\alpha} \right| \leq C_{3,t}
	\label{eq:summandBound}.
	\end{equation}
	This implies that
	\[
	\sum_{i=N_{3,t,n}}^{M_{3,t,n}} \left( \frac{i}{n} \right)^{1-\alpha} S_t(n,i)
	\leq C_{3,t} \sum_{i=N_{3,t,n}}^{M_{3,t,n}} S_t(n,i) \leq C_{3,t}
	\]
	Using  this together with our induction hypothesis,
	we may obtain from \eqref{eq:lambdamiddle} that
	\begin{equation}
	\lambda_{n,t}:=\left( \sum_{i=N_{3,t,n}}^{M_{3,t,n}} \left( \frac{i}{n} \right)^{1-\alpha} S_t(n,i) z^{\alpha}_{\infty, t+\delta t} \right)
	+ O(n^{-\frac{1}{2}}).
	\label{eq:lambdazinfinity}
	\end{equation}	
	{\bf Step 3.} In order to apply the bound \eqref{eq:bhattacharya},
	we define a Bernoulli random variable $X_{i,t}$ which takes the value $1$ if
	the $i$-th individual survives at time $t$ and $0$ otherwise. Thus
	$S_t(n,i)$ is the probability that $\sum_{j=1}^n X_{j,n}=i$.
	We define scaled random variables $\tilde{X}_{j,t}$ 
	of mean $0$ and standard deviation $1$ by
	\[
	\tilde{X}_{j,t} = \frac{X_j - s_t}{\sqrt{s_t(1-s_t)}},
	\]
	and so the appropriately scaled partial
	sum $Q_n$ is given by
	\[
	Q_n = \frac{1}{\sqrt{n}} \sum_{i=1}^n \tilde{X}_n = \frac{1}{\sqrt n}
	\sum_{j=1}^n \frac{X_j - s_t}{\sqrt{s_t(1-s_t)}}
	=
	\frac{ \left(\sum_{i=1}^n X_i\right) - n s_t}{\sqrt{n s_t(1-s_t)}}.
	\]
	We now wish to rewrite \eqref{eq:lambdazinfinity} as an integral.
	\[
	\lambda_{n,t} = 
	\int
	\id_{[N_{3,t,n},M_{3,t,n}]}(i) \left( \frac{i}{n} \right)^{1-\alpha} z^{\alpha}_{\infty, t+\delta t} 
	\ed (\sum_{j=1}^n X_{j,t})(i)
	+ O(n^{-\frac{1}{2}}).
	\]
	
	We make the substitution $i=n s_t + x \sqrt{n s_t(1-s_t)}$ to find
	\begin{multline}
	\lambda_{n,t} = 
	\int \Big[
	\id_{[N_{3,t,n},M_{3,t,n}]}(n s_t + x \sqrt{n s_t(1-s_t)}) \\
	\times \left( \frac{n s_t + x \sqrt{n s_t(1-s_t)}}{n} \right)^{1-\alpha} z^{\alpha}_{\infty, t+\delta t} \Big]
	\, \ed Q_n(x)
	+ O(n^{-\frac{1}{2}}).
	\end{multline}
	We may rewrite this as
	\begin{equation}
	\lambda_{n,t} = 
	\int \id_{[{\ell_{n,t}},{u_{n,t}}]}(x)
	\left( \frac{n s_t + x \sqrt{n s_t(1-s_t)}}{n} \right)^{1-\alpha} z^{\alpha}_{\infty, t+\delta t}
	\, \ed Q_n(x)
	+ O(n^{-\frac{1}{2}})
	\label{eq:lambdaasintegral}
	\end{equation}
	where
	\[
	\ell_{n,t}:= (N_{3,t,n} - n s)/\sqrt{n s(1-s)}, \quad
	u_{n,t}:= (M_{3,t,n} - n s)/\sqrt{n s(1-s)}.
	\]
	From our expressions for $N_{3,t,n}$ and $M_{3,t,n}$ one
	readily sees that $\ell_{n,t}$ tends to $-\infty$ at a rate proportional to $O(-\sqrt{n})$
	as $n\to\infty$. Likewise $u_{n,t}$ tends to $+\infty$ at a rate $O(\sqrt{n})$ as $n\to\infty$. We will assume that $n$ is large enough to ensure that $\ell_{n,t} <0 <u_{n,t}$.
	
	Let us define $g$ by
	\begin{equation}
	g = \id_{[{\ell_{n,t}},{u_{n,t}}]}(x)
	\left( \frac{n s_t + x \sqrt{n s_t(1-s_t)}}{n} \right)^{1-\alpha}.
	\end{equation}
	By \eqref{eq:summandBound}, $g$ is bounded by a constant independent of $n$.
	Hence $\omega_g(\R)$ is bounded independent of $n$. We can bound the derivative of $g$ inside the interval $({\ell_{n,t}},{u_{n,t}})$, independent of $n$. Hence for
	any $x \in (\frac{1}{2}{\ell_{n,t}},\frac{1}{2}{u_{n,t}})$ and
	for sufficiently large $n$, $\omega_g(B_{c n^{-\frac{1}{2}}}(x))< C_{4,t} n^{-\frac{1}{2}}$
	for some constant $C_{4,t}$ independent of $n$. It follows that
	\begin{equation}
	\int \id_{[\frac{1}{2}{\ell_{n,t}},\frac{1}{2}{u_{n,t}}]}(x) \omega_g(B_{c n^{-\frac{1}{2}}}(x)) \, \ed \Phi(x) = O(n^{-\frac{1}{2}}).
	\label{eq:oscbound1}
	\end{equation}
	Since $\ell_{n,t}$ tends to $-\infty$ at a rate proportional to $O(-\sqrt{n})$, since $g$ is bounded, and since the normal distribution
	has super-exponential decay in the tails, we have
	\begin{equation}
	\int \id_{(-\infty,\ell_{n,t}]} \omega_g(B_{c n^{-\frac{1}{2}}}(x)) \, \ed \Phi(x) = O(n^{-\frac{1}{2}})
	\label{eq:oscbound2}
	\end{equation}
	and similarly
	\begin{equation}
	\int \id_{[u_{n,t},\infty)} \omega_g(B_{c n^{-\frac{1}{2}}}(x)) \, \ed \Phi(x) = O(n^{-\frac{1}{2}}).
	\label{eq:oscbound3}
	\end{equation}
	Estimates \eqref{eq:oscbound1}, \eqref{eq:oscbound2}, \eqref{eq:oscbound3}
	together with the bound on $\omega_g(\R)$ allow us to apply the Central Limit Theorem estimate \eqref{eq:bhattacharya} to \eqref{eq:lambdaasintegral}. We note that Cram\'er's condition holds. The result is
	\begin{equation*}
	\lambda_{n,t} = 
	\int \id_{[{\ell_{n,t}},{u_{n,t}}]}(x)
	\left( \frac{n s_t + x \sqrt{n s_t(1-s_t)}}{n} \right)^{1-\alpha} z^{\alpha}_{\infty, t+\delta t}
	\, \ed \Phi(x)
	+ O(n^{-\frac{1}{2}})
	\end{equation*}
	
	We now apply Laplace's method to estimate this Gaussian
	integral (see Proposition 2.1, page 323 in \cite{ss03} ) and obtain
	\begin{equation*}
	\lambda_{n,t} = 
	s_t^{1-\alpha}
	+ O(n^{-\frac{1}{2}})
	\label{eq:lambdaasgaussian}
	\end{equation*}
	From the definition of $\tilde{\theta}$ in equation \eqref{eq:deftildetheta} and our definition of $\lambda_{n,t}$
	in \eqref{eq:defnlambdant} we obtain
	\begin{equation*}
	\tilde{\theta}_{n,t} =  \beta^\frac{1}{\rho} \exp( \xi \, \delta t) 
	s_t^\frac{{1-\alpha}}{\alpha} z_{\infty,t+\delta t} + O(n^{-\frac{1}{2}}).
	\end{equation*}
	We may now compare this with the definition of $\theta_t$ given in \eqref{eq:defphi} for the infinitely collectivised case $C=1$.
	We see that in this case
	\begin{equation*}
	\tilde{\theta}_{n,t} = \theta_t + O(n^{-\frac{1}{2}}).
	\end{equation*}
	It now follows from the recursion relations for $z_{n,t}$
	and $z_{\infty,t}$ (given in \eqref{eq:valueFunction4}
	and \eqref{eq:valueFunction4Tilde} respectively) together with our induction hypothesis
	that 
	\begin{equation*}
	z_{n,t} = z_{\infty,t} + O(n^{-\frac{1}{2}}).
	\end{equation*}
	
	This completes the induction step and the proof.
\end{proof}

\section{Numerical Results}

We illustrate our results with some numerical examples. We will restrict our attention
to the case of von Neumann--Morgernstern preferences in this section. We refer the reader
to the numerical results of \cite{ab-main} where we give numerical results for more
general homogeneous Epstein--Zin preferences. In that paper we also compare the results with those obtained using exponential Kihlstrom--Mirman preferences.

For ease of comparison with \cite{ab-main} (and other pension models based
on \cite{obr2019}) we choose the parameter values given in Table \ref{table:parameterSummaryEZ}.
Due to the positive homogeneity of our model, the choice of value for $X_0$ is unimportant.

\begin{table}[h!]
	\begin{center}		
		\begin{tabular}{ll} \toprule
			Parameter & Value \\ \midrule			
			$r_{\CPI}$ & $0.02$ \\
			$r$ & $0.047-r_{\CPI}$ \\
			$\mu$ & $0.082-r_{\CPI}$ \\
			$\sigma$ & $0.15$ \\
			$\rho$ & $-1$  \\ \bottomrule			
		\end{tabular}		
	\end{center}
	\caption{Summary of parameters used in this model}
	\label{table:parameterSummaryEZ}
\end{table}

The mortality distribution we use is for women retiring at age 65 in 2019. We
obtained this distribution using the model ``CMI\_2018\_F [1.5\%]'' as described
in \cite{cmi2018}.

We define the {\em annuity equivalent} value of each investment-consumption approach. We define this to be the price of an annuity which would
give the same gain. We define the {\em annuity outperformance} by
\[
\text{annuity outperformance} := \frac{\text{annuity equivalent}}{\text{budget}}-1.
\]
This gives a measure of the performance of the strategy relative to an annuity of
the same cost.

\subsection{Dependence on the number of individuals, $n$}

In Figure
\ref{fig:impactOfN} we show how the annuity outperformance depends upon the number of individuals
$n$ in the collective. This plot shows that as few as $40$ individuals are required
to obtain most of the benefits of collectivisation. One does not need a large fund to benefit from a collective investment: simply sharing risk with one's partner brings a substantial benefit.

\begin{figure}[htb]
	\centering
	\includegraphics{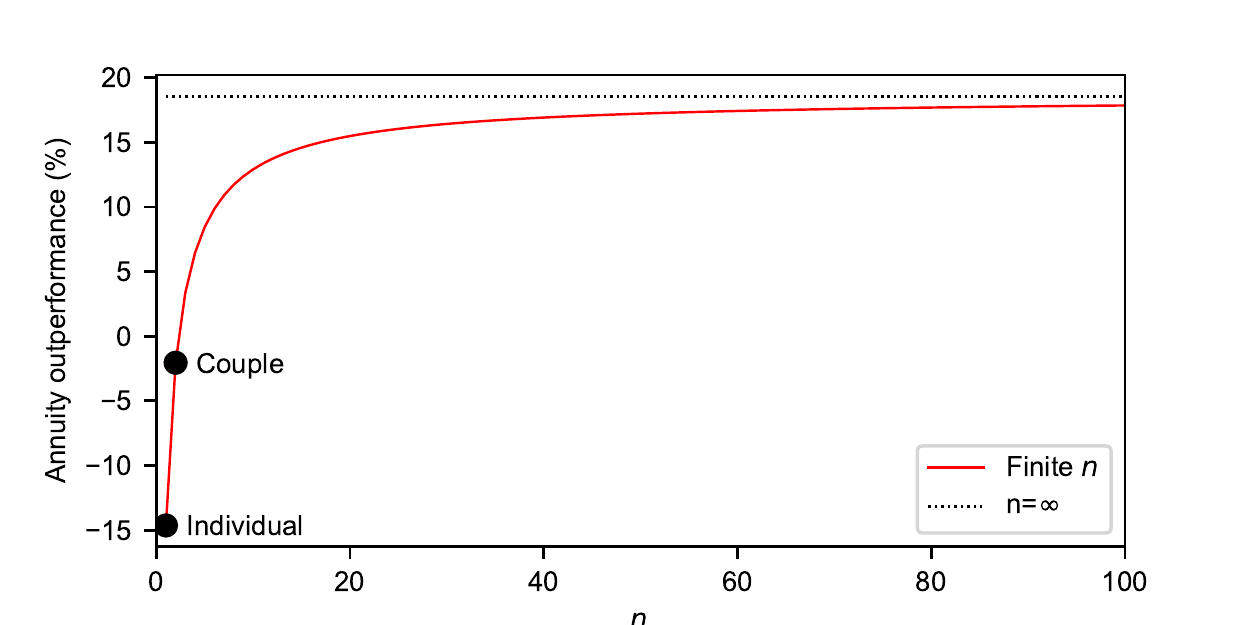}
	\caption{Dependence of annuity outperformance on the number of individuals in the collective, $n$.
		Calculation performed with von Neumann--Morgernstern preferences}
	\label{fig:impactOfN}
\end{figure}

\subsection{Dependence on market parameters}

To compare the relative impact of investment in the stock market,
inter-temporal substitution and collectivisation we have computed the
annuity outperformance for a number of different fund and market scenarios.
The results are shown in Table \ref{table:marketScenarios}. Except
where the table indicates a difference, the parameter values are described
in the previous section.

\begin{table}
	\begin{center}
		\begin{tabular}{rrrrr}
			\toprule	
			Scenario & $\mu$ & $r$ & $n$ & Annuity outperformance \\ \midrule
			1 & $0.062$ & $0.027$ & $\infty$ & $59.1\%$ \\
			2 & $0.062$ & $0.027$ & $1$      & $20.5\%$ \\
			3 & $0.027$ & $0.027$ & $\infty$ & $1.3\%$ \\
			4 & $0.000$ & $0.000$ & $\infty$ & $0\%$
		\end{tabular}
	\end{center}
	\caption{Annuity outperformance in a variety of scenarios}
	\label{table:marketScenarios}
\end{table}

If scenarios $A$ and $B$ have an annuity outperformance of $r_A$ and $r_B$ then we will say
that scenario $A$ gives an improvement of
\[
\frac{1+r_A}{1+r_B} - 1.
\]
over scenario $B$.

So comparing Scenario 1 with Scenario 2 we see that the impact of collectivisation in
this example is an improvement of $32\%$. Comparing Scenarios $1$ and $3$,
the impact of investing in stock, rather than just bonds, is even more significant,
yielding an almost $57\%$ improvement. Comparing Scenarios $3$ and $4$, we see
that exploiting inter-temporal substitution alone yields a relatively modest $1.3\%$
improvement.

\bibliography{collectivization}
\bibliographystyle{plain}

\end{document}

%% file: pensions-macros.tex
\usepackage{amsthm}
\usepackage{amsmath}
\usepackage{amssymb}
\usepackage{enumerate}
\usepackage{xcolor}
\usepackage{listings}
\usepackage{multirow}
\usepackage{color}
\usepackage{hyperref}
\usepackage{doi}
\usepackage{booktabs}
\usepackage{rotating}
\usepackage{bm}
\usepackage{tablefootnote}
\usepackage[normalem]{ulem}
\usepackage{array}
\usepackage{longtable}

\newcommand{\E}{\mathbb{E}}
\newcommand{\ed}{\mathrm{d}}
\newcommand{\R}{\mathbb{R}}
\newcommand{\id}{{\bf 1}}
\renewcommand{\P}{\mathbb{P}}

\newcommand{\EIS}{\mathrm{EIS}}
\newcommand{\CPI}{\mathrm{CPI}}

\newcommand{\DT}{\ed {\cal T}(t)}

\newtheorem{theorem}{Theorem}

\newtheorem{corollary}[theorem]{Corollary}

\newcolumntype{L}[1]{>{\raggedright\let\newline\\\arraybackslash\hspace{0pt}}p{#1}}

\theoremstyle{definition}
\newtheorem{definition}[theorem]{Definition}
\DeclareMathOperator{\Var}{Var}

\numberwithin{equation}{section}
\numberwithin{theorem}{section}

%% file: epstein-zin.bbl
\begin{thebibliography}{10}

\bibitem{armstrongClassification}
John Armstrong.
\newblock Classifying markets up to isomorphism.
\newblock {\em arXiv preprint arXiv:1810.03546}, 2018.

\bibitem{ab-main}
John Armstrong and Cristin Buescu.
\newblock Collectivised pension investment.
\newblock {\em arXiv preprint arXiv:1909.12730}, 2019.

\bibitem{ab-exponential}
John Armstrong and Cristin Buescu.
\newblock Collectivised pension investment with exponential
  {K}ihlstrom--{M}irman preferences.
\newblock {\em arXiv preprint arXiv:1911.02296}, 2019.

\bibitem{bhattacharya}
R.~N. Bhattacharya.
\newblock Rates of weak convergence and asymptotic expansions for classical
  central limit theorems.
\newblock {\em The Annals of Mathematical Statistics}, pages 241--259, 1971.

\bibitem{hall}
Robert~E Hall.
\newblock Intertemporal substitution in consumption.
\newblock {\em Journal of Political Economy}, 96(2):339--357, 1988.

\bibitem{merton1969lifetime}
Robert~C Merton.
\newblock Lifetime portfolio selection under uncertainty: The continuous-time
  case.
\newblock {\em The Review of Economics and Statistics}, pages 247--257, 1969.

\bibitem{cmi2018}
{Mortality Projections Committee}.
\newblock {\em Working Paper 119. CMI Mortality Projections Model: CMI\_2018}.
\newblock Continuous Mortality Investigation, Institute and Faculty of
  Actuaries, 2019.

\bibitem{obr2019}
{Office of Budget Responsibility (OBR)}.
\newblock {\em Supplementary forecast information release: Long-term economic
  determinants – March 2019}.
\newblock Office of Budget Responsibility, 2019.

\bibitem{pham}
Huy{\^e}n Pham.
\newblock {\em Continuous-time stochastic control and optimization with
  financial applications}, volume~61.
\newblock Springer Science \& Business Media, 2009.

\bibitem{ss03}
Elias~M Stein and Rami Shakarchi.
\newblock {\em Complex analysis}, volume~2.
\newblock Princeton University Press, 2010.

\end{thebibliography}
